\documentclass[aps,prd,onecolumn,superscriptaddress,nofootinbib,11pt]{revtex4}
\usepackage{epsfig,amsfonts,amsthm}
\usepackage{amsmath,amssymb}
\usepackage[utf8]{inputenc}
\usepackage[usenames,dvipsnames]{xcolor}
\newcommand{\be}{\begin{equation}}
\newcommand{\ee}{\end{equation}}
\newcommand{\bea}{\begin{eqnarray}}
\newcommand{\eea}{\end{eqnarray}}

\newcommand{\dst}{\displaystyle}
\newcommand{\fr}[2]{\frac{{\dst #1}}{{\dst #2}}}
\renewcommand{\Re}{\mathrm{Re }}
\renewcommand{\Im}{\mathrm{Im }}

\newcommand{\lr}[1]{ \langle #1 \rangle}
\newcommand{\Tr}{\mathrm{Tr}}
\newcommand{\Z}{\mathbb{Z}}
\newcommand{\RR}{\mathbb{R}}
\providecommand{\mtrx}[1]{\begin{pmatrix} #1 \end{pmatrix}}
\newcommand{\mmatrix}[4]{ \left(\! \begin{array}{ccc}#1 & #2 \\ #3 & #4 \end{array}\!\right) }
\newcommand{\mmmatrix}[9]{ \left(\! \begin{array}{ccc}#1 & #2 & #3\\ #4 & #5 & #6\\ #7 & #8 & #9\\ \end{array}\!\right) }

\providecommand{\bs}[1]{\boldsymbol{#1}}
\providecommand{\id}{{\boldsymbol{1}}}

\newcommand{\toCP}{\xrightarrow{CP}}
\definecolor{darkgreen}{rgb}{0.0, 0.6, 0.2}
\definecolor{DARKGREEN}{rgb}{0.0, 0.6, 0.2}

\newtheorem{theorem}{Theorem}

\def\lsim{\mathrel{\rlap{\lower4pt\hbox{\hskip1pt$\sim$}}
    \raise1pt\hbox{`$<$}}}         
\def\gsim{\mathrel{\rlap{\lower4pt\hbox{\hskip1pt$\sim$}}
    \raise1pt\hbox{$>$}}}         
\newcommand{\lrpartial}{\partial^{\hspace{-6pt}\raise3pt\hbox{\small $\leftrightarrow$}}}

\providecommand{\eq}[1]{\begin{equation} #1
    \end{equation}}
\providecommand{\la}[1]{\lambda_{#1}}

\DeclareMathOperator{\diag}{\mathrm{diag}}
\DeclareMathOperator{\re}{\mathrm{Re}}
\DeclareMathOperator{\im}{\mathrm{Im}}
\usepackage[hyperfootnotes=false]{hyperref}
\hypersetup{
   colorlinks=true,       
    linkcolor=Blue,          
    citecolor=Plum,        
    filecolor=magenta,      
    urlcolor=YellowOrange           
    }
\providecommand{\ta}{\tilde{a}}
\providecommand{\tb}{\tilde{b}}

\begin{document}
\title{
{\normalsize \hfill CFTP/18-015} \\*[7mm]
Basis-invariant conditions for $\boldsymbol{CP}$ symmetry of order 4}

\author{Igor~P.~Ivanov}\thanks{E-mail: igor.ivanov@tecnico.ulisboa.pt}
\affiliation{CFTP, Instituto Superior T\'{e}cnico, Universidade de Lisboa,
Avenida Rovisco Pais 1, 1049 Lisboa, Portugal}
\author{Celso~C.~Nishi}\thanks{E-mail: celso.nishi@ufabc.edu.br}
\affiliation{Centro de Matem\'atica, Computa\c c\~ao e Cogni\c c\~ao,
Universidade Federal do ABC - UFABC,
09.210-170, Santo Andr\'e, SP, Brazil
}
\author{Jo\~{a}o~P.\ Silva}\thanks{E-mail: jpsilva@cftp.ist.utl.pt}
\affiliation{CFTP, Departamento de F\'{\i}sica,
Instituto Superior T\'{e}cnico, Universidade de Lisboa,
Avenida Rovisco Pais 1, 1049 Lisboa, Portugal}
\author{Andreas~Trautner}\thanks{E-mail: trautner@mpi-hd.mpg.de}
\affiliation{Bethe Center for Theoretical Physics und Physikalisches Institut der Universit\"at Bonn, Nussallee 12, 53115 Bonn, Germany}
\affiliation{Max-Planck-Institut f\"ur Kernphysik, Saupfercheckweg 1, 69117 Heidelberg, Germany}

\begin{abstract}
Three-Higgs-doublet models (3HDM) allow for a novel, physically distinct form 
of $CP$ invariance: $CP$ symmetry of order 4 (CP4).
Due to the large basis change freedom in 3HDM, it is imperative to recognize the presence 
of a possibly hidden CP4 in a basis-invariant way.
In the present work, we solve this problem and establish basis-invariant necessary and sufficient conditions
for a 3HDM to possess a CP4 symmetry. 
We also derive a basis-invariant criterion to decide whether or not a CP4 symmetric 3HDM 
possesses any additional $CP$ symmetry, 
as well as a criterion to decide whether or not CP4 is spontaneously broken.
\end{abstract}

\maketitle

\section{Introduction}

There has been great interest in the $CP$ properties of models
with extended scalar sectors, because they are related to
two fundamental issues in particle physics.
Firstly,
concerning the observed excess of matter over antimatter; 
one of the most important open problems in Physics.
Explaining it through baryogenesis requires the fulfillment of the
three Sakharov conditions \cite{Sakharov:1967dj}:
violation of baryon number;
violation of the $C$ and $CP$ symmetries;
and interactions out of thermal equilibrium.
In principle, these conditions could have been met at the electroweak
phase transition of the Standard Model (SM).
Unfortunately,
in the SM,
the phase transition is ineffective and the CP-violation
parameter is far too small -- see, for example,
\cite{Trodden:1998ym, Cline:2006ts}.
In contrast, even simple extensions of the
scalar sector allow for baryogenesis at the electroweak phase transition
\cite{Cohen:1991iu}.
Secondly,
experiments at CERN's LHC have identified a fundamental scalar:
the Higgs boson \cite{Aad:2012tfa, Chatrchyan:2012xdj}.
Given the generation nature of the fermion sector,
the next important question is how many fundamental scalars there are.
These two open questions place the study of the
$CP$ properties of $N$ Higgs doublet models
(NHDM) at the center of current research, for a recent review see \cite{Ivanov:2017dad}.

The study of $CP$ violation in the scalar sector has a long history,
starting with T.D.~Lee's suggestion that $CP$ might be
conserved at the Lagrangian level, but broken spontaneously
by the vacuum~\cite{Lee:1973iz}.
This was achieved in a simple two Higgs doublet model (2HDM).
Since the late 2000's, there exists a classification of all 2HDM models
which can arise from a symmetry, their vacua,
and their $CP$ properties
\cite{Ivanov:2005hg, Maniatis:2006fs, Ivanov:2006yq,
Maniatis:2007vn, Ivanov:2007de, Ferreira:2009wh,
Ferreira:2010hy, Ferreira:2010yh}.
In the 2HDM with doublets $\phi_a$, $a =1,2$,
any model with a symmetry among the scalars is automatically invariant
under a $CP$ symmetry which, in a suitable basis, takes the canonical form:\footnote{%
Here and in the following we suppress the transformation of the space-time argument,
which transforms in the usual way $(x_0,\vec{x})\toCP(x_0,-\vec{x})$ for $CP$ transformations of any order.}
\be
\phi_a \toCP \phi_a^*\,,
\label{CP2}
\ee
Clearly, this $CP$ transformation has the property that, when applied twice, it produces the identity transformation,
$(CP)^2=1$. This means that this is a transformation of order 2, and it remains so in any other basis, 
even though the transformation law may differ from \eqref{CP2}. 
We designate $CP$ transformations of order 2 by CP2.
Moreover, a model obeys CP2 if and only if there is a basis where all parameters of the potential
are real (the so-called \textit{real basis}, which is exactly the basis where \eqref{CP2} holds) \cite{Gunion:2005ja,Ivanov:2015mwl}.

The simplest symmetries one can apply to the generic NHDM
are family (also called Higgs-flavor) symmetries
\be
\phi_a \rightarrow S_{ab} \phi_b\, ,
\label{FamSym}
\ee
and generalized $CP$ symmetries (GCP),
\be
\phi_a \toCP X_{ab} \phi_b^*\, ,
\label{GCP}
\ee
where  $a,b=1,\dots,N$,
and $S$ and $X$ are generic unitary matrices.
Generalized $CP$ symmetries appeared in \cite{Lee:1966ik}.
Their explicit application for the quark sector appeared in
\cite{Bernabeu:1986fc},
and GCP in the
scalar sector was initially explored by the Vienna group
in \cite{Ecker:1981wv, Ecker:1983hz, Neufeld:1987wa}.

The form of $S$ in Eq.~\eqref{FamSym}
is basis-dependent and, through a suitable basis change,
it can be simplified into a diagonal matrix with phases along the diagonal;
$\phi_a \rightarrow \mathrm{e}^{i \alpha_a} \phi_a$.
Similarly,
the form of $X$ is basis-dependent and can be simplified by a basis change
to a block-diagonal form \cite{Ecker:1987qp,Weinberg:1995mt},
which has on the diagonal either pure phase factors or $2\times 2$ matrices of the following type:\footnote{%
The crux is that,
under a unitary basis change of
$\phi_a \rightarrow \phi_a^\prime = U_{ab} \phi_b$,
one has
$X\rightarrow X'=UXU^{\mathrm{T}}$ (\textit{not} $UXU^\dagger$).
Therefore, only symmetric matrices $X^{\mathrm{T}}=X$ may be
completely diagonalized.}
\be
\mmatrix{c_\alpha}{s_\alpha}{-s_\alpha}{c_\alpha}\quad \mbox{as in Ref.~\cite{Ecker:1987qp},}\quad \mbox{or}\quad
\mmatrix{0}{\mathrm{e}^{i\alpha}}{\mathrm{e}^{-i\alpha}}{0}\quad \mbox{as in Ref.~\cite{Weinberg:1995mt}.}\label{block}
\ee
Not all matrices $X$ can be related by a basis-change, though.
Applying a generalized $CP$-symmetry twice, one ends up with a family
transformation with the matrix $XX^*$,
which is not necessarily equal to the identity matrix.
In fact, it may happen that one would need to
apply the $CP$ transformation $k>2$ times to arrive at the identity transformation.
The minimal number $k$ required for that
is called the order of the transformation.
We will denote a $CP$ symmetry of order $k$ as CP$k$.

Of course, one can combine several family symmetries into non-abelian groups,
combine several $CP$ symmetries, or combine family and $CP$ symmetries.
Then, typically only one of the generators can be written in the simplified form,
and the others generally cannot.
Models with 2 Higgs doublets have two interesting characteristics
\cite{Ivanov:2005hg, Ivanov:2006yq, Ferreira:2009wh}:
i) there are only 6 classes of models one can obtain with symmetries;
ii) all classes obey the usual CP2 symmetry.
In particular, applying CP$k$ to the 2HDM one obtains a model
which could have been obtained by applying CP2 and some other
symmetries.
Said otherwise,
in the symmetry-constrained 2HDM classes, there is always a real basis.

The interesting result of \cite{Ivanov:2015mwl} is that this is no
longer the case in the 3HDM.
Indeed, one can build a CP4 3HDM potential which has physical CP symmetry, yet complex phases that
cannot be removed by any basis change.
Thus, this model satisfies CP4, but not CP2.
The classification of abelian symmetry groups in the
3HDM was obtained in \cite{Ferreira:2008zy, Ivanov:2011ae},
while the classification of all 3HDM symmetry-constrained models
was performed in \cite{Ivanov:2012ry, Ivanov:2012fp}.
Extensions to the 3HDM with Yukawa interactions appear
in \cite{Serodio:2013gka, Ivanov:2013bka}.

Identification of the symmetry-constrained potentials,
their vacua,
symmetry breaking, and $CP$ properties is difficult
when working directly with the scalar fields $\phi_a$.
Such studies are much simpler using the bilinear
formalism, which we introduce in section~\ref{section-bilinear}.
This formalism appeared first in the context of the 2HDM in
\cite{Nagel:2004sw}
and
\cite{Ivanov:2005hg,Maniatis:2006fs, Ivanov:2006yq, Maniatis:2007vn, Ivanov:2007de},
but it was immediately extended to the NHDM,
for example, in \cite{Nagel:2004sw}, in Appendix B of \cite{Maniatis:2006fs}, 
and in \cite{Nishi:2006tg, Nishi:2007nh, Ivanov:2010ww, Ivanov:2010wz,
Ivanov:2010zx, Nishi:2011gc, Maniatis:2014oza, Ivanov:2014doa}.
In section~\ref{section-3HDMCP} we present the impact of CP$k$
symmetries on the 3HDM.
In section~\ref{section-alignment-examples} we introduce the notion of complete alignment in the adjoint space,
and show that only few symmetry-constrained 3HDMs exhibit this property.
We also discuss there what else, in addition to the complete alignment, is needed 
to tell CP4 from other symmetry-based situations. 
Building on these observations, we present in section~\ref{section-to-CP4} our main result:
the basis-invariant necessary and sufficient conditions for the presence of $CP$ symmetry of order $4$,
in the context of the 3HDM.
In the following section we discuss additional basis-invariant properties that a CP4-symmetric 3HDM must satisfy in order to guarantee the absence of accidental symmetries, as well as a basis-invariant condition for
detecting the spontaneous breaking of CP4 symmetry. Then we draw our conclusions.
The appendices contain a description of the symmetry-constrained 3HDM based on CP4, CP4 accompanied with CP2, $S_3$, $D_4$ and $O(2)$.

\section{Bilinear formalism}\label{section-bilinear}

\subsection{Orbit space}

Let us start with a brief recapitulation of the bilinear formalism,
with specific application to 3HDM \cite{Ivanov:2010ww,Maniatis:2014oza}.
We work with $N=3$ Higgs doublets $\phi_a$, $a=1,2,3$, all having the same electroweak quantum numbers.
The most general renormalizable 3HDM potential can be compactly written as
\be
V=Y_{ab} (\phi^{\dagger}_a\phi_b)+Z_{abcd} (\phi^{\dagger}_a\phi_b)(\phi^{\dagger}_c\phi_d)\,.\label{potential}
\ee
We construct the following $1+8$ gauge-invariant bilinear combinations $(r_0, r_i)$:
\be
r_0 = {1\over\sqrt{3}}\phi^{\dagger}_a\phi_a\,,\quad r_i = \phi^{\dagger}_a (t^i)_{ab}\phi_b\,,\quad
i=1,\dots,8\,.\label{bilinears}
\ee
Here, $t_i = \lambda_i/2$ are generators of the $SU(3)$ algebra satisfying
\be\label{eq:SU3}
[t_i,t_j] = i f_{ijk} t_k\,,\quad \text{and} \quad \{t_i,t_j\} = {1 \over 3}\delta_{ij}\id_{3} + d_{ijk} t_k\,,
\ee
with the $SU(3)$ structure constants $f_{ijk}$ and the fully symmetric $SU(3)$ invariant tensor $d_{ijk}$.
With the usual choice of basis for the Gell-Mann matrices $\lambda_i$, these have the non-zero components
\be
f_{123} = 1\,, \quad
f_{147} = -f_{156} = f_{246} = f_{257} = f_{345} = -f_{367} = {1\over 2}\,,\quad
f_{458} = f_{678} = {\sqrt{3} \over 2}\,,
\label{tensor-fijk}
\ee
as well as
\bea
&&d_{146} = d_{157} = - d_{247} = d_{256} = {1\over 2}\,,\qquad
\phantom{-} d_{344} = d_{355} = - d_{366} = - d_{377} = {1\over 2}\,,\nonumber\\
&& d_{118} = d_{228} = d_{338} = - d_{888} = {1\over \sqrt{3}}\,,\qquad
d_{448} = d_{558} = d_{668} = d_{778} = -{1\over 2\sqrt{3}}\,.\label{tensor-dijk}
\eea
Group-theoretically, $r_0$ is an $SU(3)$ singlet and $r_i$ realizes the adjoint representation of $SU(3)$.
The coefficient in the definition of $r_0$ is not fixed by this construction. We use here the definition
borrowed from \cite{Ivanov:2010ww} but alternative normalization factors are possible \cite{Maniatis:2014oza}.
The exact choice is inessential here.
In the Gell-Mann basis, the bilinears $r_i$ have the following form:
\bea
&&r_1 + i r_2 = \phi_1^\dagger \phi_2\,,\quad
r_4 + i r_5 = \phi_1^\dagger \phi_3\,,\quad
r_6 + i r_7 = \phi_2^\dagger \phi_3\,,\nonumber\\
&&
r_3 = \fr{1}{2}(\phi_1^\dagger\phi_1 - \phi_2^\dagger\phi_2)\,,\quad
r_8 = \fr{1}{2\sqrt{3}}(\phi_1^\dagger\phi_1 + \phi_2^\dagger\phi_2 - 2 \phi_3^\dagger\phi_3)\,.\label{ri-explicit}
\eea
The real vectors $r$ obtained in this way do not fill the entire real eight-dimensional space $\RR^8$ 
(the \textit{adjoint space}, whose vectors will be denoted as $x$),
but a 7D manifold in it, which is called the orbit space. The points of this space are in one-to-one correspondence
with gauge orbits within the Higgs fields space $\phi_a$.
Algebraically, the orbit space is defined by the following (in)equalities \cite{Ivanov:2010ww}:
\be
r_0 \ge 0\,, \quad r_0^2 - r_i^2 \ge 0\,,\quad
d_{ijk}r_ir_jr_k + {1\over 2\sqrt{3}}r_0(r_0^2 - 3 r_i^2) = 0\,.
\ee
A basis change in the space of Higgs doublets $\phi_a \to U_{ab} \phi_b$ with $U \in SU(3)$
leaves $r_0$ unchanged and induces an $SO(8)$ rotation of the vector $r_i$.
Not all $SO(8)$ rotations can be obtained in this way;
they must conserve, in addition, $d_{ijk}r_i r_j r_k$.

\subsection{Constructions in the adjoint space}\label{section-adjoint}

The map from the gauge invariants $\phi_a^\dag\phi_b$
to $r_i$ defined by \eqref{bilinears} is invertible.
It can be used to link an arbitrary vector $a$ in the adjoint space $\RR^8$
with a traceless hermitian $3\times 3$ matrix $A$ via
\be
A = 2 a_i t_i\, ,
\ee
so that
\be
a_i = \Tr(At_i)\, .
\ee
When working in the adjoint space, one has at one's disposal three invariant tensors $\delta_{ij}$, $f_{ijk}$,
and $d_{ijk}$, which allow one to define $SU(3)$-invariant products.
Since the space $\RR^8$ is in one-to-one correspondence with the space of traceless hermitian $3\times 3$ matrices,
we provide, for completeness, a brief ``dictionary'' between the two spaces:
\bea
c_k = 2 f_{ijk} a_i b_j & \leftrightarrow & C = -i [A, B]\,,\label{ckdk}\\
d_k = 2 d_{ijk} a_i b_j & \leftrightarrow & D = \{A, B\} - {2 \over 3}\Tr(AB) \id_{3}\,.
\eea
The main benefit of working in the adjoint space is that the Higgs potential (\ref{potential}) is a quadratic form:
\be
V = M_0 r_0 + M_i r_i + \Lambda_{0} r_0^2 + L_i r_0r_i + \Lambda_{ij} r_i r_j\,.\label{V2}
\ee
There is a one-to-one map between the parameters $Y_{ab}$ and $Z_{abcd}$ of the original potential and the coefficients
$M_0$, $M_i$, $\Lambda_0$, $L_i$, and $\Lambda_{ij}$ in the potential (\ref{V2}).
All basis-invariant structural properties of the 3HDM scalar sector are encoded in the magnitudes and the relative orientations
of the above objects. It is this geometric picture that turned out to be extremely revealing in the 2HDM
\cite{Maniatis:2006fs, Ivanov:2006yq, Maniatis:2007vn, Ivanov:2007de}.

The geometric content of $\Lambda$ requires special attention.
Within the 2HDM, all rotations of the adjoint space $\RR^3$ can be realized as $SU(2)$ transformations in the space of
Higgs doublets, that is, by Higgs-basis rotations.
Therefore, the real symmetric $3 \times 3$ matrix $\Lambda$ can be always diagonalized by an appropriate basis change.
Thus, in such a basis, $\Lambda$ can always be specified based on its three eigenvalues.

In contrast, for a generic 3HDM it is generally not possible to diagonalize the arbitrary symmetric matrix $\Lambda$
by a Higgs-basis change, simply because not all $SO(8)$ rotations of the adjoint space can be generated by $SU(3)$ Higgs-basis rotations.
In addition, $\Lambda$ of the 3HDM contains a hidden vector which can be extracted by
a contraction with $d_{ijk}$.
Group-theoretically, the $36$ independent entries of $\Lambda$ transform as $\left(8 \otimes 8\right)_S$,
which decomposes into irreducible representations (irreps) as $1 \oplus 8 \oplus 27$. The projectors $P^{ij}{}_{i'j'}$ onto these three irreps are:
\be
P_1^{ij}{}_{i'j'} = {1 \over 8} \delta_{ij}\delta_{i'j'}\,,
\quad
P_8^{ij}{}_{i'j'} = {3 \over 5} d_{ijk}d_{ki'j'}\,,
\quad
P_{27}^{ij}{}_{i'j'} = {1\over 2}(\delta_{ii'}\delta_{jj'} + \delta_{ij'}\delta_{ji'})
- P_1^{ij}{}_{i'j'} - P_8^{ij}{}_{i'j'}\,.
\ee
Thus, $K_i = d_{ijk}\Lambda_{jk}$ is a vector in the adjoint space which transforms as an octet under $SU(3)$.
It can happen that $d_{ijk}\Lambda_{jk} = 0$, in which case the octet is absent in this decomposition.
However, further octets can appear as $d_{ijk}(\Lambda^2)_{jk}$, or in higher powers $p$ of $\Lambda$, $d_{ijk}(\Lambda^p)_{jk}$.
Even if all these octets vanish, $\Lambda$ can contain non-trivial content which can never be converted into a vector.

For later use, let us define several infinite series of vectors in the adjoint space which can be constructed from $M$, $L$, 
and $\Lambda$:
\bea
&& M^{(n)}_i:=(\Lambda^n)_{ij}M_j\,,\qquad L^{(n)}_i:=(\Lambda^n)_{ij}L_j
\label{list-ML} \\
&&K^{(n)}_i:=d_{ijk}(\Lambda^n)_{jk}\,,\qquad
K^{(m,n)}_i:=(\Lambda^m)_{ii'}d_{i'jk}(\Lambda^n)_{jk}\,,
\label{list-K}
\eea
where integers $n$ and $m$ start from zero.
Since they belong to $\RR^8$, only up to 8 among them
can be linearly independent.
In specific symmetry-constrained cases, the number of linearly independent vectors can be smaller.
Finally, using invariant tensors $f_{ijk}$ and $d_{ijk}$, one can construct further vectors.

\subsection{Self-alignment of \texorpdfstring{$\boldsymbol{x_8}$}{x8}}\label{section-self-alignment}

Consider a vector $a$ in the adjoint space. Its corresponding hermitian matrix $A$ can always be diagonalized by a basis change in the fundamental space, which, back in the adjoint space,
implies that $a$ is brought to the $(x_3,x_8)$ subspace.

Now consider the \textit{star product} based on $d_{ijk}$ which is defined by
\be
(a,b) \mapsto (a\ast b)_i:=\sqrt{3}d_{ijk}a_j b_k.
\ee
Remarkably, the star product preserves the $(x_3,x_8)$ subspace via a non-linear action. That is, if vectors $a$ and $b$ have non-zero components only in the $(x_3,x_8)$ subspace,
the same is true for $c=a\ast b$, as can easily be verified using the explicit components of $d_{ijk}$ listed in \eqref{tensor-dijk}.

Here, we use the star product to define the \textit{self-alignment} property of a vector $a$.
Having rotated $a$ to the $(x_3,x_8)$ subspace, one finds that also $c=a\ast a$
lies in the same subspace with components
\be
c_3 = 2 a_3 a_8\,, \quad c_8 = a_3^2 - a_8^2\,.
\ee
In the polar coordinates on the $(x_3,x_8)$ plane, this action preserves the norm of unit vectors
and acts on the angular variable as $\alpha \mapsto \pi/2 - 2\alpha$.
Hence, the three directions $\alpha = \pi/2, \pi/6$, and $5\pi/6$ are stable under this action (cf.\ \cite{Ivanov:2010ww} for more details on this construction).
The first direction corresponds to $a$ being aligned with $x_8$, while the other two directions can be brought to this form by an allowed basis change in the fundamental space (cyclic permutation of three doublets).

Let us denote the property of a direction being stable under the action of $\ast$ as \textit{self-alignment}.
We conclude that if a vector $a$ enjoys the self-alignment property, which can be checked in any basis,
then there exists a Higgs-basis in which $a$ is aligned with the direction $x_8$.

\section{Forms of \texorpdfstring{$\boldsymbol{CP}$}{CP}-symmetry in 3HDM \label{section-3HDMCP}}

Recall that a general $CP$ transformation in NHDM acts on the Higgs doublets
as $\phi_a \toCP X_{ab} \phi_b^*$, with {a} unitary matrix $X$.
Focusing specifically on the 3HDM scalar sector,
one can classify all $CP$ transformations into
four kinds. Each kind leads to models with different symmetry content.

\subsection{CP2}\label{subsection-CP2}

If the $CP$ transformation is of order 2 ($XX^* = \id$),
then there exists a basis (called the real basis) in which $X$ is the unit matrix.
In this basis, the $CP$ transformation takes the standard form: $\phi_a \toCP \phi_a^*$.
The necessary and sufficient condition for the potential (\ref{potential}) to be explicitly CP2-conserving
is that there is a basis in which all coupling coefficients are real.

In the adjoint space, the standard $CP$ transformation leaves invariant all vectors in the 5D subspace
$V_+ = (x_3,\, x_8,\, x_1,\, x_4,\, x_6)$
and flips the sign of all vectors in the 3D subspace $V_- = (x_2,\, x_5,\, x_7)$.
Therefore, a 3HDM potential is explicitly CP2-invariant if and only if there exists a basis in which the vectors $M, L \in V_+$
and $\Lambda$ is block-diagonal with a $5\times 5$ block in $V_+$ and a $3\times 3$ block in $V_-$.

The challenge then is to determine in a basis-invariant way whether a given potential indeed has this form in some basis. The necessary and sufficient basis-invariant algebraic conditions for the existence of a real basis in the 3HDM were formulated in \cite{Nishi:2006tg} in terms of eigenvectors of the matrix $\Lambda$.

\subsection{CP4}\label{subsection-CP4}

A $CP$ transformation of order 4 is a transformation $\phi_a \toCP X_{ab} \phi_b^*$,
whose matrix $X$, in an appropriate basis, takes the form\footnote{%
Unlike in \cite{Ivanov:2015mwl} and subsequent papers, here we assign the diagonal entry to $\phi_3$ in accordance with
Gell-Mann's matrices $\lambda_i$, which also single out the third component.}
\be
X = \mmmatrix{0}{1}{0}{-1}{0}{0}{0}{0}{1}\,.
\label{X-CP4}
\ee
In this basis, CP4 acts on the adjoint space as
\bea
&& x_8 \to x_8\,,\quad (x_1, x_2, x_3) \to -(x_1, x_2, x_3)\nonumber\\
&&x_4 \to x_6\,, \quad x_6 \to -x_4\,, \quad x_5 \to -x_7\,, \quad x_7 \to x_5\,.
\eea
In other words, $x_8$ stays unchanged, vectors in $(x_1,x_2,x_3)$ flip signs
(notice that this space does not coincide with $V_-$ of CP2), 
and vectors in $(x_4,x_6)$ and $(x_5,x_7)$ are rotated by $\pm\pi/2$.
For the potential to be CP4 invariant, $M$ and $L$ must be aligned with $x_8$,
while $\Lambda$ must have the block diagonal form
\be
\Lambda =\mtrx{
    \fbox{\phantom{A}}_{\,3\times 3}& 0& 0\\
    0 & \fbox{\phantom{A}}_{\,4\times 4}& 0\\
    0 & 0 & \Lambda_{88}
    },
\label{Lamij-CP4-block}
\ee
with an arbitrary $3\times 3$ block in the subspace $(x_1,x_2,x_3)$ 
and very specific correlation patterns in the $4\times 4$ block
of the $(x_4,x_5,x_6,x_7)$ subspace. We will discuss these patterns in section~\ref{section-cp4-implies}.

Detecting the presence of CP4 in a basis-invariant way for a generic 3HDM is a challenging question which we are setting out to answer in the present work.

\subsection{CP6}

$CP$ transformations of order $6$ are always equivalent to a regular CP$2$ and a $\Z_3$ family symmetry, 
defined so that the two transformations commute.
This is a consequence of the isomorphism $\Z_6 \simeq \Z_2 \times \Z_3$.
In the 3HDM, it turns out that imposing CP6 leads to another accidental symmetry of order 2 \cite{Ivanov:2011ae}.
Hence, requiring CP6 in the 3HDM results in an $S_3$ family symmetry, on top of which comes a $CP$ symmetry.
The full symmetry group then is $S_3 \times \Z_2^{CP}$.

In this sense, CP6 does not lead to a new 3HDM, as the same potential can be obtained by imposing $S_3$ and
considering its $CP$-conserving version.
Although this type of $CP$ symmetry is not the main target of our study, we will demonstrate below that
the basis-invariant features of this model in the adjoint space are very similar to the CP4 3HDM.
Therefore, we will need to investigate this model in order to find distinctions between the two.

\subsection{Higher-order \texorpdfstring{$\boldsymbol{CP}$}{CP} symmetry}

One can also build multi-Higgs models based on even higher order $CP$-symmetries.
Within 2HDMs or 3HDMs, however, imposing invariance under CP$k$ with $k > 6$ leads to
continuous accidental family symmetries in addition to the usual CP2.
The only way to impose CP$k$ on NHDMs without producing any accidental symmetries is to take $k$ to be a power of 2 and use more than three doublets~\cite{Ivanov:2018qni}.
Again, the Higgs potential in these models can be constructed explicitly
in a suitable basis,
but the basis-invariant conditions for the presence of
higher-order $CP$-symmetries are unknown.
We will not pursue this issue in the present work.

\section{Complete alignment in adjoint space: examples
\label{section-alignment-examples}}

In this section, we show that several symmetry-based 3HDMs display a remarkable structural feature in the adjoint space: \textit{complete alignment} of all vectors.
This refers to the situation where all possible vectors in the adjoint space, ---
such as $M^{(n)}$, $L^{(n)}$, $K^{(n)}$, and $K^{(m,n)}$, as defined in \eqref{list-ML} and \eqref{list-K},
as well as their
arbitrarily complicated contractions via tensor-networks of $d_{ijk}$, $f_{ijk}$, 
and $\delta_{ij}$ leaving one index uncontracted --- are all parallel to each other.
Notice that complete alignment implies, in particular, self-alignment of any vector.
In this definition, we assume that at least one of the basic vectors $M,L,K^{(n)}$ is nonzero.

\subsection{Necessary conditions for complete alignment}

Let us begin with some straightforward criteria which must be satisfied in
order for a model to exhibit complete alignment.

Assuming complete alignment and using the arguments of section~\ref{section-self-alignment}, one can immediately establish that
all vectors must, in an appropriate basis, belong to the $x_8$ subspace, e.g.\ 
\be
M = (0, \dots, 0, M_8)\,,\qquad L = (0, \dots, 0, L_8)\,. \label{ML-alignment}
\ee
For complete alignment to be realized, these vectors must be eigenvectors of $\Lambda$.
Therefore,
the presence of any non-zero vector leads to the following block-diagonal structure of $\Lambda$:
\be
\Lambda =\mtrx{
    \fbox{\phantom{A}}_{\,7\times 7}& \bs{0}\\
     \bs{0} & \Lambda_{88}
    }\, .
\label{Lamij-step2}
\ee
Next, the vector $K_i = d_{ijk} \Lambda_{jk}$, if non-zero, must also be aligned with $r_8$,
which means that $K_i = 0$ for $i = 1, \dots, 7$. These conditions constrain the $7\times 7$ block.
With the explicit expressions for $d_{ijk}$ in \eqref{tensor-dijk} and using the notation $\{ij\} \equiv \Lambda_{ij}$,
we deduce the following list of constraints:
\bea
K_1 = 0 &\quad\Rightarrow\quad & \phantom{-}\{46\}+\{57\} = 0\,,\nonumber\\
K_2 = 0 &\quad\Rightarrow\quad &  -\{47\}+\{56\} = 0\,,\nonumber\\
K_3 = 0 &\quad\Rightarrow\quad &  \phantom{-}\{44\}+\{55\} - \{66\} - \{77\} = 0\,,\label{space123}
\eea
and
\bea
K_4 = 0 &\quad\Rightarrow\quad &  \{16\}-\{27\} + \{34\} = 0\,,\nonumber\\
K_5 = 0 &\quad\Rightarrow\quad &  \{17\}+\{26\} + \{35\} = 0\,,\nonumber\\
K_6 = 0 &\quad\Rightarrow\quad &  \{14\}+\{25\} - \{36\} = 0\,,\nonumber\\
K_7 = 0 &\quad\Rightarrow\quad &  \{15\}-\{24\} - \{37\} = 0\,.\label{space4567}
\eea
Notice that the $3\times 3$ block within the subspace $(x_1,x_2,x_3)$ is completely unconstrained, and so are the elements $\{45\}$ and $\{67\}$.
For completeness, it is also useful to collect $2\sqrt{3}K_8=2(\{11\}+\{22\}+\{33\})-2\{88\}-\{44\}-\{55\}-\{66\}-\{77\}$.

We refer to eqs.~\eqref{ML-alignment}, \eqref{Lamij-step2}, \eqref{space123},
and \eqref{space4567}
as the minimal set of necessary conditions that a model must satisfy in order to exhibit complete alignment.
Satisfying these minimal conditions does not guarantee that $K^{(n)}_i = d_{ijk}(\Lambda^n)_{jk}$ 
for all higher powers of $n$ are aligned with the other vectors.
Imposing conditions \eqref{space123}, \eqref{space4567} to $K^{(2)}$, $K^{(3)}$, etc. leads to a system of coupled algebraic equations
on the entries of $\Lambda$ which we were unable to solve.

Instead of deriving these algebraic constraints explicitly,
let us first check several symmetry-based 3HDMs in order to see which of them give rise to
complete alignment. This exercise will eventually lead us to the 
necessary and sufficient basis-invariant conditions for existence of a CP4 symmetry.

\subsection{Abelian groups are insufficient}

We begin with abelian symmetry groups which were classified for
3HDMs in \cite{Ferreira:2008zy, Ivanov:2011ae}.
To be precise, we consider here groups with 1D irreps.

All of these models are based on subgroups of the maximal abelian family symmetry group $U(1) \times U(1)$.
It suffices here to check the case with the maximal symmetry, because if a parameter relation
does not hold in the maximally symmetric $U(1) \times U(1)$ case it will certainly not hold for any of its subgroups.

In the basis where the generators of the abelian symmetries are given by rephasing transformations,
the quadratic part of the potential can only contain terms 
$\phi_a^\dagger \phi_a~(a=1,2,3)$ with arbitrary coefficients $m_{aa}^2$.
In the adjoint space, these terms lead to non-zero and generally independent (i.e.\ not aligned with any special directions)
$M_3$ and $M_8$. The same observation applies to the vector $L$. Clearly, this violates the alignment conditions.

It is irrelevant for this argument whether or not an additional ordinary $CP$ symmetry is present,
because it does not constrain the $\phi_a^\dagger \phi_a$ terms.
A generalized $CP$ which mixes doublets may lead to such constraints, but it will also produce models with higher-dimensional irreps.
We conclude that 3HDMs based on groups with 1D irreps do not generically lead to alignment.

\subsection{CP4 implies complete alignment}\label{section-cp4-implies}

From a purely group-theoretical point of view, CP4 generates the cyclic group $\Z_4$.
However, the action of CP4 on the 3 Higgses cannot be fully diagonalized by a unitary basis transformation that conserves hypercharge.
Thus it makes sense to speak of $\phi_1$ and $\phi_2$ as forming a 2D irrep of CP4.
This is the smallest group featuring a 2D irrep in 3HDMs.

The CP4-symmetric 3HDM potential written in the basis where CP4
mixes the first and second doublets \`a la \eqref{X-CP4} is given in Appendix~\ref{section-CP4-3HDM},
where we also discuss simplifications through further basis-changes.
In this basis, the vectors $M$ and $L$ satisfy the minimal necessary conditions \eqref{ML-alignment}.
The matrix $\Lambda$ has the block-diagonal form \eqref{Lamij-CP4-block}
with an arbitrary\,%
\footnote{If desired, one can use the remaining $SO(2)$ reparametrization freedom generated by $t_2$ to eliminate some entries without disrupting the $4\times 4$ block structure.}
$3\times 3$ block in the subspace $(x_1,x_2,x_3)$ and with a $4\times 4$ block which can be brought to the form
\be
\fbox{\phantom{A}}_{\,4\times 4} \to
\mathrm{diag}(\lambda_4 + |\lambda_6|,\ \lambda_4 - |\lambda_6|,\ \lambda_4 + |\lambda_6|,\ \lambda_4 - |\lambda_6|)\,,
\label{block4x4-simplest-text}
\ee
by Higgs-basis changes.
With this structure, the minimal necessary conditions
in eqs.~\eqref{space123} and \eqref{space4567} are satisfied.
Furthermore, the higher powers of $\Lambda$ keep the same block-diagonal structure
and satisfy the alignment conditions to all orders.
Thus, all vectors in eqs.~\eqref{list-ML} and \eqref{list-K} are aligned.

Finally, picking any number of vectors among $M^{(n)}$, $L^{(n)}$, $K^{(m,n)}$
and contracting them in an arbitrarily complicated way via any network of invariant tensors made out of
$d_{ijk}$, $f_{ijk}$, and $\delta_{ij}$ will never give rise to a non-aligned vector.

The proof of this statement relies on a crucial feature of any simple Lie algebra:
any tensor network with loops can be written as a linear combination of tree-level invariant tensors
(this is called primitiveness assumption in \cite{Cvitanovic:2008zz}).
In the tree-level network, one can start with the outermost branches of the form $d_{ijk} a_j b_k$,
where $a$ and $b$ are any of the above vectors.
Since they both lie in the $x_8$ subspace, so does their contraction with $d_{ijk}$.
One continues this branch-cutting procedure to arrive at the conclusion that the only possibly non-zero component of the uncontracted index is $8$. No $f_{ijk}$ can appear in a non-vanishing tree-level contraction of this kind,
simply due to its anti-symmetric nature (all appearing external vectors are already aligned in the $8$ direction).

The overall conclusion for the CP4 3HDM is the following: all adjoint-space vectors
that one can possibly construct are completely aligned.

\subsection{\texorpdfstring{$\boldsymbol{D_4}$}{D4} implies complete alignment}

The 3HDM allows one to implement the symmetry group
$D_4 \simeq \Z_4 \rtimes \Z_2$.
It has been proven in~\cite{Ivanov:2011ae} that imposing a $\Z_4$ symmetry automatically leads to explicit $CP$ conservation.
Therefore, the \mbox{$D_4$-symmetric} 3HDM is also $CP$-conserving, with an order-2 $CP$ symmetry.
However, since the total symmetry group of the model is $D_4 \times \Z_2^{CP}$ it also includes a conserved order-4 $CP$ transformation.
Therefore, the \mbox{$D_4$-symmetric} 3HDM can be viewed as a particular case of the CP4 3HDM.
Since the alignment property certainly is not lost if the symmetry is enhanced, we find that also the
$D_4$ model features complete alignment.
In the real basis of $D_4$, the action of CP4 is given exactly by \eqref{X-CP4}; see appendix \ref{ap:o(2)} for explicit generators and the form of $\Lambda$.

\subsection{\texorpdfstring{$\boldsymbol{S_3}$}{S3} implies complete alignment}

Consider now a 3HDM with an $S_3$ family symmetry and with explicit $CP$-violation.
The general potential of the $S_3$ 3HDM 
is shown in the $\Z_3$ diagonal basis in Appendix~\ref{appendix-S3}
or the real basis in Appendix~\ref{ap:o(2)}.
The vectors $M$ and $L$ again satisfy the minimal necessary conditions \eqref{ML-alignment},
while the matrix $\Lambda$ takes the form:
\be
\Lambda = \left(\begin{array}{ccc|cccc|c}
a & \cdot & \cdot & c & s & c & s & \cdot\\
\cdot & a & \cdot & s & -c & -s & c & \cdot\\
\cdot & \cdot & b & \cdot & \cdot & \cdot & \cdot & \cdot\\ \hline
c & s & \cdot & d & \cdot & f & g & \cdot\\
s & -c & \cdot & \cdot & d & g & -f & \cdot\\
c & -s & \cdot & f & g & d & \cdot & \cdot \\
s & c & \cdot & g & -f & \cdot & d & \cdot \\ \hline
\cdot & \cdot & \cdot & \cdot & \cdot & \cdot & \cdot & h
\end{array}\right)\,,
\label{Lam-S3}
\ee
where dots indicate zero entries.
It is straightforward to verify that all conditions \eqref{space123} and \eqref{space4567} are satisfied,
showing that $K_i$ is aligned with $x_8$.
What is more important is that the pattern in \eqref{Lam-S3} reproduces itself
in all powers of $\Lambda$.
Therefore, also all vectors $K^{(n)}_i$ are aligned with $x_8$.
By the same arguments as used for the CP4 3HDM,
we conclude that this model features complete alignment, despite not having a CP4 transformation.

We stress that this model is generally $CP$-violating. Enforcing the symmetry group $S_3 \times Z_2^{CP}$
is equivalent to setting $s = 0$.\footnote{%
If both $s=0$ and $c=0$ the potential has a continuous symmetry.}
In this case, the subspace $(x_2,x_5,x_7)$ decouples from the rest, as was discussed in Section~\ref{subsection-CP2},
but it has no effect on alignment.
As mentioned before, we could have also arrived at this case by imposing a single generalized $CP$-symmetry CP6.
We conclude, therefore, that complete alignment is not sufficient for the existence of a CP4 symmetry.

\subsection{Distinguishing CP4 from \texorpdfstring{$\boldsymbol{S_3}$}{S3}}\label{subsection-distinguishing}

The conclusion of the previous subsection implies that complete vector alignment in the adjoint space
cannot, by itself, single out the CP4 3HDM models. One needs more basis-invariant information
to distinguish it from the $S_3$ 3HDM.

We prove here that the eigensystem of $\Lambda$ readily offers these criteria.
There are two versions of checks: using only eigenvalues and using eigenvectors.\footnote{The fact that we cannot diagonalize $\Lambda$ by the basis change
in the space of doublets $\phi_a$ does not matter.
We are not claiming here that we can bring $\Lambda$ to the diagonal form.
We are just saying that any real symmetric matrix can always be expanded via its eigensystem.}

We saw that $\Lambda$ of the CP4 3HDM has a completely generic symmetric $3\times 3$ block,
a very constrained $4\times 4$ block with two pairs of eigenvalues $\lambda_4 \pm \lambda_6$, and the $\Lambda_{88}$
entry. Thus, the eigenvalue degeneracy pattern of a \textit{generic} CP4 3HDM is $1+1+1+1+2+2$.
In the case of $CP$-violating $S_3$ 3HDM, 
the eigenvalues of (\ref{Lam-S3}) always come with the degeneracy pattern $1+1+2+2+2$,
which is best seen in a different basis where generators of $S_3$ are real, see Eq.~\eqref{Lam-S3-real}.
Thus, if, in addition to the complete alignment,
we observe the eigenvalue degeneracy pattern $1+1+1+1+2+2$,
we immediately conclude the presence of CP4.

This criterion allows us to detect a generic CP4.
It may happen that a valid CP4 model has some accidental
degeneracy among its eigenvalues, which would prevent the application of this criterion.
To cope with these cases,
we propose to look at the eigenvectors of $\Lambda$,
which are definitely different for CP4 and $S_3$ models,
regardless of the eigenvalues.

Indeed, the subspace $(x_1,x_2,x_3)$ is in a special position with respect to the direction $x_8$.
Take two adjoint space vectors $a$ and $b$ such that $a \in (x_1,x_2,x_3)$ and $b$ is along $x_8$.
Then, the vectors are $f$-orthogonal:
$F_i \equiv f_{ijk} a_j b_k = 0$.
Conversely,
when $b$ is along $x_8$ and $a$ is perpendicular to $b$
($a_i b_i=0$)
then $F_i =0$ implies that $a$ belongs to $(x_1,x_2,x_3)$.

One can check both statements, using the structure constants in Eq.~\eqref{tensor-fijk}.
Alternatively,
one can use the map of Section~\ref{section-adjoint} to construct the corresponding
traceless hermitian $3\times 3$ matrices $A$ and $B$.
They have the following structure:
\be
A = \mtrx{*&*&0\cr*&*&0\cr 0&0&0}\,, \qquad
B \propto \mtrx{1&0&0\cr 0&1&0\cr 0&0&-2}\,.\label{AB-commute}
\ee
These matrices commute and $\Tr(A B) = 0$.
Conversely,
if a traceless matrix commutes with $B$,
and it has no piece proportional to $B$,
then it must be of the form $A$.

Thus, we arrive at the following basis-invariant distinction between the CP4 and $S_3$ 3HDMs:
although both possess an eigenvector along $x_8$, only the CP4 3HDM possesses three other mutually orthogonal eigenvectors
which are both orthogonal and $f$-orthogonal to it. This criterion resolves the ambiguity.

\subsection{Groups with triplet representations}
\label{G:3D}

We have seen from examples that complete alignment follows from groups with doublet representations, including CP4; see remarks in section~\ref{section-cp4-implies}.
All these groups lie inside $SU(3)\rtimes \Z_2^{CP}$.
If we consider realizable groups with triplet representations, all of them containing either $A_4$ or $\Delta(27)$\,\cite{Ivanov:2014doa}, we know that 
\eq{
M_i=L_i=K_i=0\,.
}
Indeed, it is easy to see from the branching rules that there are no invariants of these groups within the adjoint of $SU(3)$ and a nonzero vector contracted with $r_i$ 
would not be invariant. Now, invariance of the potential \eqref{V2} implies 
\eq{
\label{invariance}
D_{ij}(g)M_j=M_i,\quad
D_{ij}(g)L_j=L_i,\quad
D_{ii'}(g)D_{jj'}(g)\Lambda_{i'j'}=\Lambda_{ij},
}
for any group element $g$ acting through a representation $D$ on vectors.
The fact that all vectors transform in the same way as $M_i,L_i,K_i$ together with the invariance properties \eqref{invariance} imply that any vector $F_i=F_i(M,L,\Lambda)$ built from the basic quantities of the potential is also invariant under the group.
Thus, all vectors in in Eqs.\,\eqref{list-ML} and \eqref{list-K} vanish as well, a fact that can also be checked explicitly.
Therefore, imposing invariance under a group with triplet representation leads to constraints stronger than complete alignment 
and so we will not consider these groups further in this paper.

%
%


\section{Detecting a CP4 symmetry}\label{section-to-CP4}

In the previous section we showed that the CP4 3HDM leads to the complete alignment and,
in addition, the eigenvectors and eigenvalues of $\Lambda$ possess certain characteristic properties.
Now, based on these results, we prove the converse statements, namely, that if certain
basis-invariant properties are satisfied, the model possesses a CP4 symmetry.

We will give two versions of these conditions.
First, we will formulate and prove the main Theorem, which unambiguously
detects the presence of a CP4 symmetry in all cases where it is present.
Checking these necessary and sufficient conditions requires, in addition to the complete alignment,
verification that the eigenvectors of the matrix $\Lambda$ satisfy certain properties.
Then, we will show a simplified version of these conditions,
which involve the eigenvalues but not the eigenvectors of $\Lambda$
and, therefore, may be computationally less expensive.
These simplified conditions can detect a CP4 symmetry in a \textit{generic} situation
but will miss the CP4 symmetry at certain special points in parameter space.

\subsection{Necessary and sufficient conditions for a CP4 symmetry}

\begin{theorem}
Consider the vectors $M$, $L$, and the matrix $\Lambda$ defined in \eqref{V2}.
Compute the eigenvectors of $\Lambda$.
The model possesses a CP4 symmetry if and only if all of the following conditions are satisfied:
\begin{itemize}
\item
there exists an eigenvector of $\Lambda$, denoted $v^{(8)}$ which is self-aligned in the sense of section~\ref{section-self-alignment}:
$d_{ijk}v^{(8)}_j v^{(8)}_k$ is parallel to $v^{(8)}_i$;
\item
there exist exactly three other mutually orthogonal
eigenvectors of $\Lambda$, denoted $v^{(\alpha)}$ with $\alpha = 1,2,3$,
which are $f$-orthogonal to $v^{(8)}$: $f_{ijk}v^{(8)}_j v^{(\alpha)}_k = 0$;
\item
the vectors $M$, $L$, $K_i = d_{ijk}\Lambda_{jk}$, and
$K^{(2)}_i = d_{ijk}(\Lambda^2)_{jk}$, if non-zero,
are parallel to $v^{(8)}$.
\end{itemize}
\end{theorem}

\begin{proof}
{\bf Step 1.}
As we explained in section~\ref{subsection-CP4},
the existence of a CP4 symmetry implies, among other, that,
in a suitable basis, the matrix $\Lambda$ takes the block-diagonal form \eqref{Lamij-CP4-block},
with subspaces $(x_8)$ and $(x_1, x_2, x_3)$ decoupled from each other and from the rest.
These two conditions can be formulated in a basis-invariant way and checked in any basis.

First, using the results of section~\ref{section-self-alignment}, the self-aligned eigenvector $v^{(8)}$ 
can be always pointed by a basis change along the direction $x_8$. 
In this basis, the matrix $\Lambda$ takes the form \eqref{Lamij-step2}.
Second, using the results of section~\ref{subsection-distinguishing} on eigenvectors,
we conclude that the three mutually orthogonal eigenvectors of $\Lambda$
which are $f$-orthogonal to $v^{(8)}$ can only belong to the $(x_1, x_2, x_3)$ subspace.
Thus, the other four eigenvectors lie in the $(x_4, x_5, x_6, x_7)$ subspace,
and one arrives at the desired block-diagonal structure \eqref{Lamij-CP4-block}.
Higher powers of the matrix $\Lambda$ also possess this block-diagonal form. 

{\bf Step 2.}
The $4\times 4$ block in the $(x_4, x_5, x_6, x_7)$ subspace must exhibit a certain pattern 
in order to be compatible with CP4 symmetry.
This pattern can be fixed by the conditions that the first seven components of 
the vectors $K_i = d_{ijk}\Lambda_{jk}$ and $K^{(2)}_i = d_{ijk}(\Lambda^2)_{jk}$ are zero,
or in other words that these two vectors, even if non-zero, lie in the $x_8$ subspace.

Indeed, the conditions \eqref{space4567} are automatically satisfied, 
while the conditions \eqref{space123} shape the $4\times 4$ block to the following form:
\be
a \cdot \id_4 +
\mtrx{
b & d & c & -s \\
d & -b & -s & -c \\
c & -s & b' & d' \\
-s & -c & d' & -b'}\,.\label{block4x4}
\ee
All 7 free parameters here are independent.
Since no further constraint follows from the vector $K$, we consider $K^{(2)}$.
The corresponding matrix $\Lambda^2$ also has the block-diagonal form with the $4\times 4$ block consisting of
the two matrices already shown in \eqref{block4x4} together with the new contribution
\be
\mtrx{
b^2 + d^2 + c^2 + s^2 & 0 & c(b+b') - s(d+d') & -c(d-d') - s(b-b') \\
\cdot & b^2 + d^2 + c^2 + s^2  & c(d-d') +s(b-b') & c(b+b') - s(d+d') \\
\cdot  & \cdot  & b^{\prime 2} + d^{\prime 2} + c^2 + s^2  & 0 \\
\cdot  & \cdot  & \cdot & b^{\prime 2} + d^{\prime 2} + c^2 + s^2 }\,.
\ee
where we use dots to denote the symmetric entries below the diagonal for clarity.
Applying the same constraints \eqref{space123} to this matrix, we obtain the conditions
\be
b^2 + d^2 = b^{\prime 2} + d^{\prime 2}\,, \quad
c(b+b') = s(d+d')\,,\quad
s(b-b') = -c(d-d')\,.\label{bcd-conditions}
\ee
They amount to new independent conditions.
If we define
\eq{
\la{6}\equiv b-id\,,~~\la{7}\equiv b'-id'\,,~~\la{5}\equiv c+is\,,
}
the conditions \eqref{bcd-conditions} are recast in the form
\be
|\lambda_6| = |\lambda_7|\,, \quad 
\re[\la{5}^*(\la6+\la7)]=0\,,\quad
\im[\la{5}^*(\la6-\la7)]=0\,.
\ee
The last two conditions can be combined into $\lambda_5^*\lambda_6 = - \lambda_5\lambda_7^*$,
which implies the relation between $\psi_{5,6,7}$
(the arguments of $\lambda_{5,6,7}$):
$\psi_6+\psi_7=2\psi_5 + \pi$. 
These constraints on the arguments and absolute values coincide 
with \eqref{conditions567} of the CP4 3HDM given in Appendix~\ref{section-CP4-3HDM}.
The $4\times 4$ block \eqref{block4x4} takes the same form as
\eqref{block4x4-full}, and overall we recover the matrix $\Lambda$
exactly of the same type as in the CP4 3HDM.
Since the $4\times 4$ block of the form \eqref{block4x4-full} can always be brought to the diagonal form \eqref{block4x4-simplest-text}, all higher-power vectors $K^{(n)}$, if non-zero, are also aligned with $x_8$. 
Finally, if the vectors $M$ and $L$ are non-zero, they are also aligned with $x_8$ and therefore do not spoil the CP4 invariance.
Thus, the model indeed possesses a CP4 symmetry and the proof is complete.
\end{proof}

\subsection{Detecting a generic CP4}

Checking the necessary and sufficient conditions formulated in Theorem 1 requires determination
of the full eigensystem of the matrix $\Lambda$.
However, in most cases, the presence of a CP4 symmetry can be deduced already from the complete alignment
and the eigenvalues of $\Lambda$, without computation of eigenvectors.
This statement comes from the observation made in section~\ref{subsection-distinguishing}
that the eigenvalues of the CP4 3HDM exhibit the degeneracy pattern $1+1+1+1+2+2$,
which is impossible in the other completely aligned 3HDM without CP4, the $S_3$ 3HDM.
Let us now make this statement precise.

\begin{theorem}[Generic CP4]
Consider vectors $M$, $L$, and the matrix $\Lambda$ defined in \eqref{V2}.
If the vectors $M$, $L$, and $K^{(n)}$ with $1 \le n \le 7$, at least one of which is nonzero, 
respect complete alignment, and if, in addition, $\Lambda$ has four non-degenerate eigenvalues, then the model has a CP4 symmetry.
\end{theorem}

\begin{proof}
As in Theorem 1, one first needs to establish the block-diagonal structure \eqref{Lamij-CP4-block}.
However, since we do not explicitly rely on the eigenvectors of $\Lambda$, the proof proceeds differently.

{\bf Step 1.} Take a nonzero vector among $M$, $L$, and $K^{(n)}$.
For definiteness we assume $K_i = d_{ijk}\Lambda_{jk} \neq 0$ but any will do.%
\footnote{A situation where all vectors vanish can be symmetry protected only by a symmetry inside $SU(3)$ with a three-dimensional representation;
see section~\ref{G:3D}.}
Since \textit{complete} alignment implies \textit{self-}alignment for all vectors,
$K$ is self-aligned and we use the arguments of Section~\ref{section-self-alignment}
to conclude that, after an appropriate basis change, $K$ can be made to lie exclusively in the $x_8$ subspace.
Furthermore, complete alignment implies that $K$ is an eigenvector of $\Lambda$, which leads us to the block-diagonal structure \eqref{Lamij-step2}.
By assumption, all vectors $K^{(n)}$ are also  aligned with $x_8$, and therefore the entries of $\Lambda^n$ for all $n$ must satisfy conditions \eqref{space123} and \eqref{space4567}.

{\bf Step 2.} Now we prove that, under the assumptions of this theorem, 
the $7\times 7$ block splits into $3\times 3$ and $4\times 4$ blocks as in \eqref{Lamij-CP4-block}.
Let us write $\Lambda$ within this 7D subspace via eigenvalues
and eigenvectors:
\be
\Lambda_{ij} = \sum_{\alpha} \Lambda_\alpha e_{i}^{(\alpha)} e_{j}^{(\alpha)}\,.\label{eigensystem}
\ee
Here, $\alpha$ runs over all eigenvalues, even if some of them are zero.
Now, it may happen that some eigenvalues $\Lambda_\alpha$ are degenerate with multiplicities $m_\alpha > 1$.
In order to take that into account, let us rewrite \eqref{eigensystem} as
\be
\Lambda_{ij} = \sum_{\alpha} \Lambda_\alpha P^{(\alpha)}_{ij}\,, \quad
P^{(\alpha)}_{ij} = \sum_{k_\alpha = 1}^{m_\alpha}e_{i}^{(k_\alpha)} e_{j}^{(k_\alpha)}\,.\label{eigensystem2}
\ee
Now the first summation runs over all \textit{distinct} eigenvalues $\Lambda_\alpha$,
while the second summation runs over all eigenvectors corresponding to this eigenvalue.
The total number $p$ of distinct eigenvalues within the 7D subspace is at most $7$
and at least 4 by assumption.
The matrices $P^{(\alpha)}_{ij}$ are the projectors on the corresponding subspaces;
they satisfy
\be
P^{(\alpha)} P^{(\beta)} = \delta_{\alpha\beta} P^{(\alpha)}\,, \quad \sum_\alpha P^{(\alpha)} = \id_{7}\,.\label{projectors-7D}
\ee
Next, also the matrices $\Lambda$ to the power $n$ are expanded in the form \eqref{eigensystem2}
with eigenvalues $(\Lambda_{\alpha})^n$.
So, the condition that $d_{ijk}(\Lambda^n)_{jk} = 0$ within the 7D subspace $(x_1,\dots, x_7)$
means that \textit{all} linear combinations
\be
\sum_{\alpha} (\Lambda_\alpha)^n S^{(\alpha)}_k= 0\,,\quad \mbox{where} \quad
S^{(\alpha)}_k := d_{ijk} P^{(\alpha)}_{ij}\,.\label{linearS}
\ee
Writing linear combinations (\ref{linearS}) for $n = 1, \dots, p$,
we obtain each time a linear combination of $p$ vectors $S^{(\alpha)}_k$ in the 7D space.
Since all $\Lambda_\alpha$ are distinct, this implies that each individual vector $S^{(\alpha)}_k = 0$.

The above statement applies to all vectors $S^{(\alpha)}_k$ which correspond to 
non-zero eigenvalues $\Lambda_\alpha$. However, even if $\Lambda$ has a zero eigenvalue,
which we denote as $\Lambda_0 = 0$, the corresponding vector $S^{(0)}_k = 0$, too.
Indeed, since the projectors sum up to $\id_7$,
and since $d_{ijk}\delta_{ij} = 0$ within the 7D subspace,
we get:
\be
S^{(0)}_k =  d_{ijk} P^{(0)}_{ij} = d_{ijk} \left(\delta_{ij} - \sum_{\alpha \not = 0} P^{(\alpha)}_{ij}\right)
= d_{ijk} \delta_{ij} - \sum_{\alpha \not = 0} S^{(\alpha)}_k = 0\,.
\ee

The essence of the above trick deserves emphasis:
Instead of constraining a generic matrix $\Lambda$ and its powers,
we constrain their eigenspace projectors $P^{(\alpha)}$, for which taking powers has no effect.

Now, consider a non-degenerate eigenvalue $\Lambda_1$.
Then $P^{(1)}_{ij} = e_{i}^{(1)} e_{j}^{(1)}$, and we are looking for solutions of
$e_{i}^{(1)} e_{j}^{(1)}d_{ijk} = 0$ within the 7D subspace.
We can solve this set of equations via the same Eqs.~(\ref{space123}) and (\ref{space4567}), where each entry $\{ij\}$ is now understood as the product of the two components of the same eigenvector, $e_i e_j$. 
Solving simultaneously the three conditions (\ref{space123}),
we conclude that the components $e_{4,5,6,7} = 0$, while $e_{1,2,3}$ are unconstrained.
Thus, an eigenvector corresponding to a non-degenerate eigenvalue (in the 7D subspace) must lie within the subspace $(x_1,x_2,x_3)$ and nowhere else.

Therefore, if we require the full matrix $\Lambda$ to possess four non-degenerate eigenvalues,
this can only happen if one of the corresponding four eigenvectors lies in $x_8$ 
and the other three belong to the $(x_1,x_2,x_3)$ subspace.
The remaining four eigenvectors, by orthogonality, then must belong to the $(x_4,x_5,x_6,x_7)$ subspace,
and they must have at least pairwise degenerate eigenvalues.
Thus, we arrive at the split block-diagonal form \eqref{Lamij-CP4-block}.

{\bf Step 3} goes exactly as step $2$ of Theorem 1 and completes the proof.
\end{proof}

Theorem 2 proposes sufficient conditions for a 3HDM to contain a CP4 symmetry.
One first needs to check that the vectors $M$, $L$, $K^{(n)}$ are self-aligned and parallel,
and that the common direction is an eigenvector of $\Lambda$.
Then, one needs to compute the eigenvalues of $\Lambda$.
If there are four non-degenerate eigenvalues, we detect the presence of a CP4 symmetry.

If there are fewer than four non-degenerate eigenvalues, this method fails as it can 
miss a valid CP4-symmetric model.
This algorithm also fails in the case when $M$, $L$, $K^{(n)}$
are all zero vectors,
since in this case it is impossible to identify the self-aligned eigenvector without actually 
computing the full eigensystem. However, this only happens at exceptional isolated points in the parameter space 
or when a symmetry with triplet representation is present.
For a generic scan in the parameter space, the sufficient conditions still represent a useful check.

\section{Discussion and conclusions}\label{section-conclusions}

\subsection{Subtleties with degenerate eigenvalues}

Theorem 1 does not only present the necessary and sufficient conditions for a 3HDM to possess
a CP4 symmetry, but also proposes a concrete algorithm which can be employed in any basis.
However, when implementing it, one may face a technical difficulty when the model 
has degeneracy among eigenvalues beyond the generic $1+1+1+1+2+2$ pattern.

First, $\Lambda_{88}$ may be degenerate with other eigenvalues of $\Lambda$.
In this case it may happen that only one direction out of the entire corresponding eigenspace
satisfies the self-alignment property. In order not to miss a valid CP4 symmetry in such situations,
one may need to parametrize the vectors of this eigenspace and check if any of them exhibits self-alignment.

Next, there may exist more than one direction in the eigenspace corresponding to a degenerate eigenvalue exhibiting the self-alignment property.
One can pick up any of them, denote it as $v^{(8)}$, and then search for three other
mutually orthogonal eigenvectors which would be both orthogonal and $f$-orthogonal to $v^{(8)}$.
However if we fail to find such triplet of eigenvectors, it does not yet mean that the model has no CP4 symmetry.
It may be just the wrong choice of the self-aligned direction which was associated with $v^{(8)}$.
One then would need to check all possible assignments for $v^{(8)}$. 
Only if none of them leads to the desired triplet of eigenvectors we can claim that the model has no CP4 symmetry.

These complications call upon a refined concrete algorithm which would be capable of
detecting a CP4 symmetry in all possible situations of accidental degeneracies among eigenvalues.
Constructing such algorithm is delegated to a future work.
For now, we stress that these complications are just technical and do not jeopardize 
the proof of Theorem 1.

\subsection{CP4 symmetry vs. CP4 3HDM}

Theorem 1 gives the necessary and sufficient conditions for a 3HDM to possess
\textit{a} CP4 symmetry. However, in addition to CP4, a model could possess other symmetries.
The total symmetry group then would be larger and, as shown in \cite{Ivanov:2012fp}, 
it would always automatically contain a CP2 symmetry. Thus, the CP4 symmetry would lose its defining role, as
the same model could be built by imposing CP2 and an appropriate family symmetry.

If one wants to single out \textit{the} CP4 3HDM model \cite{Ivanov:2015mwl} 
where CP4 is the only symmetry, one first must check the presence 
of a CP4 symmetry and then verify the \textit{absence} of any additional CP2.
In Appendix~\ref{section-CP4-CP2}, we describe all options for extending a
CP4 symmetric 3HDM by CP2. The different options depend on whether CP2 and CP4 commute or not.
In the commuting case, the minimal enhancement of the total symmetry group is to $D_4 \times \Z_2^{CP}$,
which is also studied in Appendix~\ref{ap:o(2)}.
In the non-commuting case,
the minimal resulting symmetry is $(\Z_2 \times \Z_2)\rtimes Z_2^{CP}$.
In Appendix~\ref{appendix-recognition-CP2}, we describe a basis-invariant algorithm to distinguish 
these models from the pure CP4 model which works in both cases.
For this, one needs to construct $f$-products amongst two pairs of eigenvectors of the $(x_4, x_5, x_6, x_7)$ subspace
and check if any of them is an eigenvector of $\Lambda$. 
This feature then allows for a straightforward algorithmic implementation to distinguish between pure CP4 and higher symmetries.

The problem of basis-invariant recognition of an additional CP2 symmetry
in CP4 3HDM has recently also been tackled in \cite{Haber:2018iwr}. 
Starting with a CP4 3HDM and assuming that the CP4 symmetry is unbroken,
the authors discovered a basis-invariant ${\cal N}$ in the form of a high-degree polynomial
of the quartic coefficients of the potential and the vacuum expectation values of the doublets.
This invariant is zero if and only if the model possesses an additional CP2 symmetry that commutes with CP4.
In Appendix~\ref{appendix-recognition-CP2}, we derive an algorithm,
which is short, transparent, covers both commuting and non-commuting cases,
and does not rely on vacuum expectation values.

\subsection{Spontaneous breaking of CP4}

Suppose the presence of a CP4 symmetry together with the absence of any additional CP2 symmetry is detected 
in a basis-invariant way.
Then, there exits an immediate basis-invariant criterion to decide 
whether a chosen vector of vacuum expectation values, $\lr{\phi_a}=v_a$, is CP4 conserving or not.
In particular, there is no need to explicitly reconstruct the CP4 transformation.

The criterion for CP4 conservation after minimization is that the vector $\lr{r}$ is self-aligned
and parallel to $M$, $L$, and $K$'s and is, therefore, an eigenvector of $\Lambda$.
If this property does not hold, CP4 is spontaneously broken.
This criterion follows because in the standard CP4 basis of \eqref{X-CP4} a CP4 conserving vacuum expectation value has (after appropriate rephasing) the form $\lr{\phi}=v(0,0,1)$ 
which implies that $\lr{r}$ is aligned to $x_8$.

\subsection{Conclusions}

In summary, we brought up and solved the question of basis-invariant recognition
of the presence of a CP4 symmetry in 3HDM. Since this question cannot be solved
with the traditional technique of constructing $CP$-odd basis invariants and then setting them to zero,
we developed a new approach, which makes use of the
adjoint space constructions and, in particular, the eigensystem
of the matrix $\Lambda$. 
The final result is a set of necessary and sufficient conditions for the presence of a CP4 symmetry formulated as Theorem 1.
In generic settings, the presence of CP4 can also be
determined with a computationally less expensive approach
which only requires the knowledge of eigenvalues but not eigenvectors of $\Lambda$. 
In addition, we have presented necessary and sufficient basis-invariant criteria to detect the presence of other symmetries beyond CP4, 
as well as to determine whether CP4 is spontaneously broken by the Higgs vacuum expectation value.

The presented algorithms can be implemented in parameter scans of the scalar sector of 3HDM.
In particular, they offer an efficient path to explore the intriguing phenomenology
of CP4 3HDM without the need to stay in one particular basis. 

\subsection*{Acknowledgments}
I.P.I.\ acknowledges funding from the Portuguese
\textit{Fun\-da\-\c{c}\~{a}o para a Ci\^{e}ncia e a Tecnologia} (FCT) through the FCT Investigator 
contract IF/00989/2014/CP1214/CT0004 under the IF2014 Programme,
and through the contracts UID/FIS/00777/2013, CERN/FIS-NUC/0010/2015, and PTDC/FIS-PAR/29436/2017,
which are partially funded through POCTI (FEDER), COMPETE, QREN, and the EU.
I.P.I.\ and J.P.S.\ also acknowledge the support from National Science Center, Poland, via the project Harmonia (UMO-2015/18/M/ST2/00518).
The work of A.T.\ has been supported by the German Science Foundation (DFG) within the SFB-Transregio TR33 ``The Dark Universe''.
C.C.N.\ acknowledges partial support by Brazilian funding agencies Fapesp through grant 2014/19164-6 and CNPq through grant 308578/2016-3.

\appendix

\section{CP4 3HDM potential: from the most general to the simplest}\label{section-CP4-3HDM}

Here, we summarize the results on the Higgs potential of
the CP4 3HDM \cite{Ivanov:2015mwl}
and the basis-change freedom available for its simplification. 
For a recent related study, see \cite{Haber:2018iwr}.

Given a $CP$ transformation of order 4 acting on the Higgs fields as
$\phi_a \mapsto X_{ab}\phi_b^*$ with some matrix $X$ satisfying $XX^* \not = \id$,
$(XX^*)^2 = \id$, we can always find a basis in which $X$ takes the following form:
\be
X = \mmmatrix{0}{1}{0}{-1}{0}{0}{0}{0}{\mathrm{e}^{i\beta}}\,.\label{X-CP4_2}
\ee
The general 3HDM potential invariant under CP4 with this matrix $X$
is $V = V_0 + V_{CP4}$ \cite{Ivanov:2011ae}, where
\bea
V_0 &=& - m_{11}^2 (\phi_1^\dagger \phi_1 + \phi_2^\dagger \phi_2) - m_{33}^2 \phi_3^\dagger \phi_3
+ \lambda_1 \left[(\phi_1^\dagger \phi_1)^2 + (\phi_2^\dagger \phi_2)^2\right] + \lambda_2 (\phi_3^\dagger \phi_3)^2 +
\label{V0}\\
&+&\lambda_3 (\phi_3^\dagger \phi_3) (\phi_1^\dagger \phi_1 + \phi_2^\dagger \phi_2)
+ \lambda'_3 (\phi_1^\dagger \phi_1) (\phi_2^\dagger \phi_2) + \lambda_4 \left(|\phi_1^\dagger \phi_3|^2 + |\phi_2^\dagger \phi_3|^2\right)
+ \lambda'_4 |\phi_1^\dagger \phi_2|^2\,,\nonumber
\eea
and
\be
V_{CP4} = 
\lambda_5 (\phi_1^\dagger\phi_3)(\phi_2^\dagger\phi_3)
+ {\lambda_6 \over 2} (\phi_1^\dagger\phi_3)^2 + {\lambda_7 \over 2} (\phi_2^\dagger\phi_3)^2 +
{\lambda_8 \over 2} (\phi_1^\dagger \phi_2)^2 + \lambda_9(\phi_1^\dagger\phi_2)\left(\phi_1^\dagger\phi_1-\phi_2^\dagger\phi_2\right) + h.c.\,.\label{V-CP4}
\ee
Here, $m_{11}^2$, $m_{22}^2$, and $\lambda_{1,2,3,4}$ are real while $\lambda_5$ through $\lambda_9$ can be complex.
Their phases are denoted as $\psi_{5,\dots,9}$ and their sines and cosines are denoted as $s_{5,\dots,9}$ and $c_{5,\dots,9}$.
Not all of them are independent, though. The following conditions must be met:
\be
|\lambda_6|=|\lambda_7|\,,\quad \psi_6 + \psi_7 = 2\psi_5 + \pi = - 2\beta\,.\label{conditions567}
\ee
In the adjoint space, we find that $M_i = (0,\dots, 0, M_8)$ and $L_i = (0,\dots, 0, L_8)$, with
\be
M_8 = {2\over\sqrt{3}}(m_{11}^2 - m_{33}^2)\,, \qquad L_8 = {4 \over 3}(\lambda_1 - \lambda_2) + {2 \over 3}(\lambda_3' - \lambda_3)\,.
\ee
The matrix $\Lambda$ has the block-diagonal form \eqref{Lamij-CP4-block}
with the blocks
\be
\Lambda_{88} = \fr{2\lambda_1 + 4\lambda_2 - 4\lambda_3 + \lambda_3'}{3}\,,\qquad
\fbox{\phantom{A}}_{\,3\times 3} = \left(
\begin{array}{ccc}
\lambda_4' + |\lambda_8| c_8 & \phantom{\lambda_4'}-|\lambda_8|s_8 & \phantom{-}2|\lambda_9|c_9 \\
\cdot & \lambda_4' - |\lambda_8| c_8  & -2|\lambda_9|s_9 \\
\cdot & \cdot & 2\lambda_1 - \lambda_3'
\end{array}
\right)\label{block3x3}
\ee
and
\be
\fbox{\phantom{A}}_{\,4\times 4} = \lambda_4 \cdot \id_4 + \left(
\begin{array}{cccc}
|\lambda_6|c_6 & -|\lambda_6|s_6 & \phantom{-}|\lambda_5|c_5 & -|\lambda_5|s_5 \\
\cdot & -|\lambda_6|c_6 & -|\lambda_5|s_5 & -|\lambda_5|c_5 \\
\cdot & \cdot & \phantom{-}|\lambda_6|c_7 & -|\lambda_6|s_7 \\
\cdot & \cdot & \cdot & -|\lambda_6|c_7
\end{array}
\right)\,,\label{block4x4-full}
\ee
where the dots below the diagonal indicate the repeated entries of a symmetric matrix.
We see that the blocks in the subspaces $x_8$ and $(x_1, x_2, x_3)$ are completely unconstrained,
while the $4\times 4$ block in the subspace $(x_4,x_5,x_6,x_7)$ depends on 5 free parameters:
$\lambda_4$, $|\lambda_5|$, $|\lambda_6|$,
as well as three phases $\psi_5$, $\psi_6$, $\psi_7$ subject to one condition \eqref{conditions567}.

The CP4 3HDM potential \eqref{V-CP4} can be simplified by the residual freedom of basis changes that preserve the matrix $X$ of \eqref{X-CP4_2} up to rephasings.\footnote{%
The corresponding groups are the $SU(2)$ that acts on the first two doublets and the $U(1)$ generated by $t_8$.}
In particular, one can set $\beta=0$, make $\lambda_6 = \lambda_7$ real and positive, and eliminate $\lambda_5$ \cite{Ferreira:2017tvy},
so that the potential \eqref{V-CP4} becomes
\be
V_{CP4} = {\lambda_6 \over 2} \left[(\phi_1^\dagger\phi_3)^2 + (\phi_2^\dagger\phi_3)^2 
+ (\phi_3^\dagger\phi_1)^2 + (\phi_3^\dagger\phi_2)^2\right] +
\left[{\lambda_8 \over 2} (\phi_1^\dagger \phi_2)^2 + \lambda_9(\phi_1^\dagger\phi_2)\left(\phi_1^\dagger\phi_1-\phi_2^\dagger\phi_2\right) + h.c.\right].\label{V-CP4-simple}
\ee
The coefficients $\lambda_8$ and $\lambda_9$ stay complex. Note that there is no basis in which all coefficients are simultaneously real \cite{Ivanov:2015mwl}.
This amounts to bringing the $4\times 4$ block to a diagonal pairwise degenerate form,
\be
\fbox{\phantom{A}}_{\,4\times 4} \to
\mathrm{diag}(\lambda_4 + \lambda_6,\ \lambda_4 - \lambda_6,\ \lambda_4 + \lambda_6,\ \lambda_4 -\lambda_6)\,,
\label{block4x4-simplest}
\ee
while the $\Lambda_{88}$ entry and the $3\times 3$ block keep their general form.
This basis choice clearly demonstrates that higher powers of $\Lambda$ feature the same block diagonal pattern:
a generic $3\times 3$ block, a generic $88$ component, and a diagonal pairwise degenerate $4\times 4$ block.
This structure is still form-invariant under the $SO(2)$ reparametrization group generated by $t_2$.

\section{Distinguishing the CP4 3HDM from models that also possess CP2}\label{section-CP4-CP2}

Let us see how the CP4 3HDM is further constrained by imposing additional symmetries and how we can
detect this in a basis invariant way.
To do this, we start working in a basis in which the CP4 symmetry is generated by the matrix $X$  in \eqref{X-CP4}.
The full classification of discrete symmetry-based 3HDMs is presented in \cite{Ivanov:2012ry, Ivanov:2012fp}.
Upon imposing additional symmetries the total symmetry group is enhanced to one of these models or to a model with continuous symmetry.

It is common to all extensions of the CP4 model that they automatically contain conserved CP2 transformations.
Therefore, starting with the CP4 3HDM, a minimal extension of the symmetry group is to just add a CP2 transformation.
The newly arising combined $CP$ transformation, \mbox{$(\mathrm{CP2})^{-1} \cdot \mathrm{CP4} \cdot \mathrm{CP2}$}, then is again of order $4$.
Requiring that there is a minimal
symmetry enhancement, it should generate the
same group as the original CP4 transformation.
There are two options to do this: the resulting order-4 $CP$ transformation
can either be the original CP4 or its inverse. We will now derive the total resulting symmetry for either case
and establish a basis-invariant criterion to detect the presence of the resulting extra symmetries.

\subsection{Commuting CP4 and CP2}

If  $(\mathrm{CP2})^{-1} \cdot \mathrm{CP4} \cdot \mathrm{CP2} = \mathrm{CP4}$,  the two $CP$ transformations commute.
In a basis where the CP4 transformation is given by the matrix $X$ in Eq.~\eqref{X-CP4}, 
the desired CP2 transformation is generated by a matrix $X_2$ that solves $X_2 X^* X_2=X$. 
The most general solution reads
\be
X_2 = \mmmatrix{\mathrm{e}^{i\alpha}\cos{\varphi}}{i \sin{\varphi}}{0}{i \sin{\varphi}}{\mathrm{e}^{-i\alpha}\cos{\varphi}}{0}{0}{0}{\pm 1}\,,\label{X-CP2_1}
\ee
with two free real parameters $\alpha$ and $\varphi$. The possible minus sign in the third component can always be removed by a global rephasing of $\phi$.
Since $X$ is invariant under all $SU(2)$ basis changes in the first two components, so is the equation that we have used to derive $X_2$.
Hence, the matrix $X_2$ must be form-invariant under these basis changes. Therefore, we can repeat the steps of Appendix~\ref{section-CP4-3HDM} and write the potential in a basis where it takes the form \eqref{V-CP4-simple} without 
changing the form of \eqref{X-CP2_1}. Within this basis, 
we then require that the CP2 transformation generated by $X_2$ should only minimally enhance the symmetry.
This means we require that it should not impose any constraint on $\lambda_6$ or equivalently, leave the $4\times 4$ block of $\Lambda$ invariant.
This requirement then restricts the free parameters of \eqref{X-CP2_1} to $\alpha,\varphi=0, \pi$.\footnote{%
Taking any other choice within the general solution for CP2 will also commute with CP4, but these choices will unavoidably enhance the symmetry even
further. Clearly it is necessary for our argument to detect the minimal possible symmetry enhancement of CP4, while the argument 
is not spoiled by the possibility of an even larger possible symmetry.}
That is, in the basis \eqref{V-CP4-simple} the minimal symmetry enhancement
is given by a CP2 transformation that corresponds to conjugation with the identity matrix.
The other possibility	 of having $X_2=\diag(-1,-1,1)$ is equivalent because this element is already contained in the group generated by CP4.

The product of the $CP$ transformations, $a_4:=\mathrm{CP2} \cdot \mathrm{CP4}$, then is a family symmetry of order 4 with
transformation matrix $X$.
Imposing this transformation on the potential \eqref{V-CP4-simple} forces the coefficients $\lambda_8, \lambda_9$ to be real.
For the matrix $\Lambda$, the reality of all coefficients implies that the $3\times 3$ block \eqref{block3x3} 
splits into a $2\times 2$ block in the subspace $(x_1, x_3)$ and the $\Lambda_{22}$ entry.
In other words, the direction $x_2$ becomes an eigenvector of $\Lambda$.
This additional special eigenvector is enough to determine the presence of a symmetry beyond CP4 in a basis invariant way as we will discuss below.

We remark that upon imposing the minimal CP2 extension here, there appears an accidental symmetry such that the 
total symmetry group of the model actually is $D_4 \times \Z_2^{CP}$ \cite{Ivanov:2011ae,Ivanov:2014doa}.
An alternative approach, therefore, would be to study the $D_4$-symmetric 3HDM
from the start and we do this in Appendix~\ref{ap:o(2)}.
We note that when the $D_4$ transformations are fixed, there is no reparametrization freedom left and in the basis of \eqref{block4x4-simplest} with standard real representation for $D_4$, $\Lambda$ is completely diagonal.

As a digression from the main line of arguments, let us review how 
the Higgs family symmetry group accidentally enlarges to $D_4$, \cite{Ivanov:2011ae,Ivanov:2014doa}.
If a 3HDM possesses a $\Z_4$ family symmetry generated by $a_4$, then there appears an accidental CP2$'$ symmetry
which \textit{does not} commute with $a_4$.
This is best seen in the basis where $a_4 = \mathrm{diag}(-i, i, 1)$ and the desired CP2$'$ is based on a diagonal matrix.
However, all $CP$ symmetries we have identified so far can be generically written as $a_4^k\cdot \mathrm{CP4}$.
All of them commute with $a_4$ and, therefore, none of them can play the role of CP2$'$. 
Back in the original basis, the desired CP2$'$ is based on the orthogonal matrix $X_2$ given below in Eq.~\eqref{X-CP2_22}.
Stripping it off the conjugation, one gets the desired symmetry $a_2$ of the same form.
Now, since $a_2^{-1} a_4 a_2 = a_4^{-1}$, they, by themselves, generate the family symmetry group $\langle a_2, a_4\rangle \simeq D_4$,
on top of which we have the standard $CP$ symmetry. Since $D_4$ is expressed in the real basis, this $CP$ commutes with it,
making the total symmetry group $D_4 \times \Z_2^{CP}$. 

\subsection{Non-commuting CP4 and CP2}

If $(\mathrm{CP2})^{-1} \cdot \mathrm{CP4} \cdot \mathrm{CP2} = (\mathrm{CP4})^{-1}$, the most general
transformation matrix of the new CP2 symmetry is given by the solution to the equation $X_2X^*X_2=X^\mathrm{T}$, which is
\be
X_2 = \mmmatrix{\mathrm{e}^{i\alpha}\cos{\varphi}}{\phantom{-}\sin{\varphi}}{0}{\sin{\varphi}}{-\mathrm{e}^{-i\alpha}\cos{\varphi}}{0}{0}{0}{\pm1}\,.\label{X-CP2_2}
\ee
Again, $X_2$ is form-invariant under the basis changes which lead to \eqref{V-CP4-simple} 
and we again decide to work in that basis.
The potential minus sign in the third component can again be removed by a global rephasing. 
The requirement of minimal symmetry enhancement then restricts $\alpha=0, \pi$, implying that the minimal additional symmetry is given by CP2 
generated by complex conjugation together with a matrix 
\be
X_2 = \mmmatrix{\cos{\varphi}}{\phantom{-}\sin{\varphi}}{0}{\sin{\varphi}}{-\cos{\varphi}}{0}{0}{0}{1}\,.\label{X-CP2_22}
\ee
In this case, $\lambda_8$ and $\lambda_9$ are constrained in such a way 
that the $3\times3$ block of $\Lambda$ acquires an eigenvector in the $(x_1,x_3)$ plane. 
By the remaining $SO(2)$ basis freedom generated by $t_2$ this eigenvector can always 
be aligned with either the $x_1$ or $x_3$ directions.
In the former case $\varphi=0$ in \eqref{X-CP2_22}, implying real $\lambda_8$ and imaginary $\lambda_9$,
while in the latter case $\varphi=\pi/2$, implying $\lambda_9 = 0$ without constraining $\lambda_8$.

No additional accidental symmetries appear in this case, and the total symmetry group 
is given by $(\Z_2 \times \Z_2)\rtimes \Z_2^{CP}$, 
where the family symmetry group $\Z_2 \times \Z_2$ is generated by $\mathrm{CP2}\cdot\mathrm{CP4}$ and $\mathrm{CP4}\cdot\mathrm{CP2}$.

\subsection{Basis-invariant recognition of an extra CP2}\label{appendix-recognition-CP2}

Summarizing the above cases of various CP2 symmetries in addition to CP4, 
we can state that all of them lead to a
simplification of the $3\times 3$ block \eqref{block3x3},
namely, (at least) one of the three directions $x_1$, $x_2$, $x_3$ becomes an eigenvector of $\Lambda$
in the symmetry basis studied above.
Such a basis is still compatible with the basis \eqref{block4x4-simplest} where the $4\times 4$ block is diagonal but the latter is defined only up to $SO(2)$ reparametrization transformations.
In one of these cases, the $D_4$ symmetry can be immediately spotted by checking whether the $3\times 3$ block has $x_2$ as an eigenvector.
For the case of $(\Z_2 \times \Z_2)\rtimes \Z_2^{CP}$, one has to check if the subspace $(x_1,x_3)$ contains an eigenvector.
We will now show how to formulate these criteria in a basis-independent way.

This can be done using the $f$-product of eigenvectors from the $(x_4, x_5, x_6, x_7)$ subspace.
Indeed, we work in the basis where $\Lambda$ is diagonal in this subspace, see Eq.~\eqref{block4x4-simplest},
so that an eigenvector $e^{(\alpha=4,5,6,7)}$ is aligned with $x_{\alpha=4,5,6,7}$.
The two subspaces with degenerate eigenvalues are $(x_4,x_6)$ and $(x_5,x_7)$.
Let us define generic eigenvectors within these two subspaces as
\be
q^{(\gamma)} = e^{(4)} \cos\gamma + e^{(6)} \sin\gamma\,,\quad
p^{(\delta)} = e^{(5)} \cos\delta + e^{(7)} \sin\delta\,.\label{pq-vectors}
\ee
First, taking two mutually orthogonal eigenvectors corresponding to the same eigenvalue 
unambiguously defines the $x_2$ direction:
\be
f_{ijk} q^{(\gamma)}_j q^{(\gamma+\pi/2)}_k = f_{ijk} p^{(\delta)}_j p^{(\delta+\pi/2)}_k \quad  \mbox{lies along $x_2$.}
\ee
Next, the $f$-product of generic $p$ and $q$ lies in the $(x_1,x_3,x_8)$ subspace:
\be
f_{ijk} q^{(\gamma)}_j p^{(\delta)}_k = 
{1\over 2}\left(\sin(\gamma+\delta),\, 0,\, \cos(\gamma+\delta),\, \dots,\, \sqrt{3}\cos(\gamma-\delta)\right)\,.
\ee
The $f$-product of the other pair of vectors
\be
f_{ijk} q^{(\gamma+\pi/2)}_j p^{(\delta+\pi/2)}_k = 
{1\over 2}\left(-\sin(\gamma+\delta),\, 0,\, -\cos(\gamma+\delta),\, \dots,\, \sqrt{3}\cos(\gamma-\delta)\right)\,.
\ee
Therefore, their sum and difference split the $(x_1,x_3)$ subspace from the $x_8$ direction:
\be
f_{ijk} q^{(\gamma)}_j p^{(\delta)}_k - f_{ijk} q^{(\gamma+\pi/2)}_j p^{(\delta+\pi/2)}_k = 
\left(\sin(\gamma+\delta),\, 0,\, \cos(\gamma+\delta),\, \dots,\, 0\right)\,.
\ee
By simultaneously varying $\gamma$ and $\delta$, one can scan all directions in this subspace and check
if any of them is an eigenvector of $\Lambda$.

Thus, a general algorithm to detect any extra symmetry beyond CP4 is the following.
Using pairs of eigenvectors from the $(x_4, x_5, x_6, x_7)$ subspace, construct the $x_2$ direction
and the $(x_1,x_3)$ subspace. If any of them contains an eigenvector of $\Lambda$,
we have an additional CP2 symmetry, and the model is not the pure CP4 3HDM.
If none of them contains an eigenvector of $\Lambda$, we have the pure CP4 3HDM.

\section{\texorpdfstring{$\boldsymbol{S_3}$}{S3}-symmetric 3HDM}\label{appendix-S3}

The $S_3$-symmetric 3HDM was first proposed back in 1978 \cite{Derman:1978nz,Pakvasa:1977in}
and has been studied in numerous papers since then \cite{Ivanov:2017dad}.
Several conventions exist to write the potential of this model. Here we stick to the notation of \cite{Ivanov:2012ry,Ivanov:2012fp,Ivanov:2014doa}
where the $\Z_3$ subgroup is diagonal. The symmetry group $S_3$ is generated by $a_3$ and $b$ with the form
\be
\label{S3:diag}
a_3 = \mmmatrix{\omega^2}{\cdot}{\cdot}{\cdot}{\omega}{\cdot}{\cdot}{\cdot}{1}\,,\qquad
b = \mmmatrix{\cdot}{1}{\cdot}{1}{\cdot}{\cdot}{\cdot}{\cdot}{1}\,,
\ee
where $\omega:=\mathrm{e}^{2\pi i /3}$. Here and for all matrices below dots indicate zero entries.
The Higgs potential is written as $V = V_0 + V_{S_3}$, with the same $V_0$ as in \eqref{V0} and
\be
V_{S_3} = \lambda_5 (\phi_1^\dagger\phi_3)(\phi_2^\dagger\phi_3)
+ \lambda_{10} \left[(\phi_2^\dagger\phi_1)(\phi_3^\dagger\phi_1) + (\phi_1^\dagger\phi_2)(\phi_3^\dagger\phi_2)\right] + h.c.
\label{VS3}
\ee
The coefficient $\lambda_5$ can be always made real, but then $\lambda_{10}$ remains, in general, complex.
Notice that the $\lambda_5$ term is the same as in the CP4 3HDM.

In the adjoint space, one has the same form of $M$ and $L$ as before,
while the matrix $\Lambda$ now takes now the form
\be
\Lambda_{ij} =
\left(
\begin{array}{ccc|cccc|c}
\lambda_4' & \cdot & \cdot & \Re\lambda_{10} & \Im\lambda_{10} & \Re\lambda_{10} & \Im\lambda_{10} & \cdot\\
\cdot & \lambda_4' & \cdot & \Im\lambda_{10} & -\Re\lambda_{10} & -\Im\lambda_{10} & \Re\lambda_{10} & \cdot\\
\cdot & \cdot & 2\lambda_1 - \lambda_3' & \cdot & \cdot & \cdot & \cdot & \cdot\\ \hline
\Re\lambda_{10} & \Im\lambda_{10} & \cdot & \lambda_4 & \cdot & \lambda_5 & \cdot & \cdot\\
\Im\lambda_{10} & -\Re\lambda_{10} & \cdot & \cdot & \lambda_4 & \cdot & -\lambda_5 & \cdot\\
\Re\lambda_{10} & -\Im\lambda_{10} & \cdot & \lambda_5 & \cdot & \lambda_4 & \cdot & \cdot \\
\Im\lambda_{10} & \Re\lambda_{10} & \cdot & \cdot & -\lambda_5 & \cdot & \lambda_4 & \cdot \\ \hline
\cdot & \cdot & \cdot & \cdot & \cdot & \cdot & \cdot & \Lambda_{88}
\end{array}\right)\,,
\label{Lamij-S3}
\ee
with $\Lambda_{88} = 2\lambda_1/3 + 4(\lambda_2 - \lambda_3 + \lambda_3')/3$.
It is straightforward to verify that the conditions \eqref{space123} and \eqref{space4567} are all satisfied, implying that $K$ is aligned with $x_8$.

\section{\texorpdfstring{$\boldsymbol{O(2)}$}{O(2)} symmetry and subgroups}
\label{ap:o(2)}

We describe here the $O(2)$ group and its various subgroups which include $S_3$ and $D_4$.
Different from appendix \ref{appendix-S3}, $S_3$ here will be given in the real basis.
The form of the vectors $M, L$ is easy to recover: $SO(2)$ invariance leads to
\eq{
M\sim L\sim (0,*,0;0,0,0,0;*)\,.
}
$O(2)$ invariance eliminates the second component and the same applies to the subgroups $S_3$ or $D_4$.

The $SO(2)$ symmetry can be generated by a transformation $\phi\mapsto\exp(it_2\theta)\phi$, with the standard Gell-Mann matrix $t_2$, recall Eq.~\eqref{eq:SU3}.
Invariance under this group constrains the quartic couplings to
\eq{
\label{L:so2}
\Lambda=
\left(
\begin{array}{ccc|cccc|c}
\Lambda_{11} & \cdot & \cdot & \cdot & \cdot & \cdot & \cdot & \cdot \\
 \cdot & \Lambda_{22} & \cdot & \cdot & \cdot & \cdot & \cdot & \Lambda_{28} \\
 \cdot & \cdot & \Lambda_{11} & \cdot & \cdot & \cdot & \cdot & \cdot \\
 \hline
 \cdot & \cdot & \cdot & \Lambda_{44} & \Lambda_{45} & \cdot & -\Lambda_{56} & \cdot \\
 \cdot & \cdot & \cdot & \Lambda_{45} & \Lambda_{55} & \Lambda_{56} & \cdot & \cdot \\
 \cdot & \cdot & \cdot & \cdot & \Lambda_{56} & \Lambda_{44} & \Lambda_{45} & \cdot \\
 \cdot & \cdot & \cdot & -\Lambda_{56} & \cdot & \Lambda_{45} & \Lambda_{55} & \cdot \\
 \hline
 \cdot & \Lambda_{28} & \cdot & \cdot & \cdot & \cdot & \cdot & \Lambda_{88} \\
\end{array}
\right)\,.
}
Further imposition of%
\eq{
\label{tb}
\tb=\diag(1,-1,1)\,,
}
enlarges $SO(2)$ to $O(2)$ and implies $\Lambda_{28}=\Lambda_{56}=0$.
The combination $SO(2)\times\Z_2^{\rm CP}$ is even stronger and additionally implies $\Lambda_{45}=0$ and we get a
diagonal structure for $\Lambda$ with degenerate eigenvalues with multiplicities $(1,2,2,2,1)$.
We have complete alignment from $O(2)$ symmetry irrespective of additional CP.

Among the $\Z_n$ symmetries in the 3HDMs, only the cases $n=2,3,4$ are realizable\,\cite{Ivanov:2011ae}.
Embedded in the $SO(2)$ above, we can use the following matrices as generators of $\Z_3$ and $\Z_4$:
\eq{
\ta_3=\left(
\begin{array}{ccc}
 -\frac{1}{2} & -\frac{\sqrt{3}}{2} & \cdot \\
 \frac{\sqrt{3}}{2} & -\frac{1}{2} & \cdot \\
 \cdot & \cdot & 1 \\
\end{array}
\right)\,,\quad
\ta_4=\mtrx{\cdot&1&\cdot\cr-1&\cdot&\cdot\cr\cdot&\cdot&1}
\,.
}
We discard $\Z_2$ because it will not lead to a nonabelian group when combined with $\tb$.
Note that $\ta_3$ is the same as $a_3$ in \eqref{S3:diag} after a change of basis $\tilde{g}=U_r gU_r^\dag$, with
\eq{
U_r=
\left(
\begin{array}{ccc}
 \frac{1}{\sqrt{2}} & \frac{1}{\sqrt{2}} & \cdot \\
 \frac{i}{\sqrt{2}} & -\frac{i}{\sqrt{2}} & \cdot \\
 \cdot & \cdot & 1 \\
\end{array}
\right)
\,.
}
The same applies to $\tb$ in \eqref{tb} and $b$ in \eqref{S3:diag}.
Combination of $\ta_3$ or $\ta_4$ with $\tb$ leads respectively to $S_3=D_3$ or $D_4$.
Invariance under these two groups leads to a $\Lambda$ structurally different from the $O(2)$ invariant one.
Invariance under $S_3=\langle\ta_3,\tb\rangle$ allows terms off the central blocks:
\eq{
\label{Lam-S3-real}
\Lambda=
\left(
\begin{array}{ccc|cccc|c}
 \Lambda_{11} & \cdot & \cdot & \cdot & \cdot & \Lambda_{16} & \Lambda_{17} & \cdot \\
 \cdot & \Lambda_{22} & \cdot & \cdot & \cdot & \cdot & \cdot & \cdot \\
 \cdot & \cdot & \Lambda_{11} & -\Lambda_{16} & -\Lambda_{17} & \cdot & \cdot & \cdot \\
 \hline
 \cdot & \cdot & -\Lambda_{16} & \Lambda_{44} & \Lambda_{45} & \cdot & \cdot & \cdot \\
 \cdot & \cdot & -\Lambda_{17} & \Lambda_{45} & \Lambda_{55} & \cdot & \cdot & \cdot \\
 \Lambda_{16} & \cdot & \cdot & \cdot & \cdot & \Lambda_{44} & \Lambda_{45} & \cdot \\
 \Lambda_{17} & \cdot & \cdot & \cdot & \cdot & \Lambda_{45} & \Lambda_{55} & \cdot \\
 \hline
 \cdot & \cdot & \cdot & \cdot & \cdot & \cdot & \cdot & \Lambda_{88} \\
\end{array}
\right)\,.
}
In contrast, invariance under $D_4=\langle\ta_4,\tb\rangle$ allows the $3\times 3$ block to have nondegenerate eigenvalues:
\eq{
\Lambda=
\left(
\begin{array}{ccc|cccc|c}
 \Lambda_{11} & \cdot & \cdot & \cdot & \cdot & \cdot & \cdot & \cdot \\
 \cdot & \Lambda_{22} & \cdot & \cdot & \cdot & \cdot & \cdot & \cdot \\
 \cdot & \cdot & \Lambda_{33} & \cdot & \cdot & \cdot & \cdot & \cdot \\
 \hline
 \cdot & \cdot & \cdot & \Lambda_{44} & \cdot & \cdot & \cdot & \cdot \\
 \cdot & \cdot & \cdot & \cdot & \Lambda_{55} & \cdot & \cdot & \cdot \\
 \cdot & \cdot & \cdot & \cdot & \cdot & \Lambda_{44} & \cdot & \cdot \\
 \cdot & \cdot & \cdot & \cdot & \cdot & \cdot & \Lambda_{55} & \cdot \\
 \hline
 \cdot & \cdot & \cdot & \cdot & \cdot & \cdot & \cdot & \Lambda_{88} \\
\end{array}
\right)\,.\label{Lam-D4-real}
}
This is clearly a particular form of the CP4 invariant block structure in \eqref{Lamij-CP4-block} with $4\times 4$ block in the form \eqref{block4x4-simplest-text}.
Canonical $CP$ is clearly a symmetry and CP4 in \eqref{X-CP4} is easily identified as $\ta_4\cdot CP$.


\end{document}